\documentclass[10pt,twocolumn]{article}
\usepackage{amssymb}
\usepackage{amsfonts} \usepackage{amsmath} \usepackage{amsthm} \usepackage{bm}
\usepackage{latexsym} \usepackage{epsfig}

\newtheorem*{theorem*}{Theorem}
\newtheorem{theorem}{Theorem}[section]
\newtheorem{lemma}[theorem]{Lemma}
\newtheorem*{proposition*}{Proposition}

\newtheorem{claim}[theorem]{Claim}
\newtheorem{corollary}[theorem]{Corollary}

\newtheorem{definition}[theorem]{Definition}
\newtheorem{remark}[theorem]{Remark}
\newtheorem{conjecture}[theorem]{Conjecture}

\newcommand{\ignore}[1]{}

\newcommand{\enote}[1]{} \newcommand{\knote}[1]{}
\newcommand{\rnote}[1]{}

\DeclareMathOperator*{\argmax}{arg\,max}
\newcommand{\cS}{{\mathcal{S}}}
\newcommand{\cA}{{\mathcal{A}}}
\newcommand{\cF}{{\mathcal{F}}}
\newcommand{\Vl}{V_{\rm \tiny loop}}

\newcommand{\tsigma}{\widetilde{\sigma}}
\newcommand{\hsigma}{\widehat{\sigma}}

\newcommand{\E}[1]{{\mathbb{E}}\left[{#1}\right]}

\renewcommand{\P}[1]{{\mathbb{P}}\left[{#1}\right]}
\newcommand{\CondP}[2]{{\mathbb{P}}\left[{#1}\middle\vert{#2}\right]}
\newcommand{\ZombP}[2]{{\mathbb{Q}}\left[{#1}\middle\vert\middle\vert{#2}\right]}
\newcommand{\Ind}[1]{\mathbf{1}\left(#1\right)}

\def\ux{\underline{x}}
\def\uy{\underline{y}}

\def\traj{\lambda}

\newcommand{\Dv}{\rho_V}
\newcommand{\De}{\rho_E}

\newcommand{\eps}{\epsilon}

\newcommand{\N}{\mathbb N} \newcommand{\R}{\mathbb R}

\newcommand{\cX}{\mathcal{X}}

\newcommand{\ind}{{\bf 1}}

\newcommand{\sgn}{\mathrm{sgn}}

\newcommand{\half}{{\textstyle \frac12}}

\renewcommand{\phi}{\varphi}

\begin{document}
\title{Efficient Bayesian Social Learning on Trees}

\author{Yashodhan Kanoria\footnote{Microsoft Research New England and
    Department of Electrical Engineering, Stanford University. Supported
    by 3Com Corporation Stanford Graduate Fellowship.}~ and
  Omer Tamuz\footnote{Microsoft Research New England and Weizmann
    Institute. Supported by ISF grant 1300/08} }

\date{\today}
\maketitle
\begin{abstract}
  We consider a set of agents who are attempting to iteratively
  learn the `state of the world' from their neighbors in a social
  network. Each agent initially receives a noisy observation of
  the true state of the world. The agents then repeatedly `vote'
  and observe the votes of some of their peers, from which they
  gain more information. The agents' calculations are Bayesian
  and aim to myopically maximize the expected utility at each
  iteration.

  This model, introduced by Gale and Kariv (2003), is a natural
  approach to learning on networks. However, it has been
  criticized, chiefly because the agents' decision rule appears
  to become computationally intractable as the number of
  iterations advances. For instance, a dynamic programming
  approach (part of this work) has running time that is
  exponentially large in $\min(n, (d-1)^t)$, where $n$ is the
  number of agents.

  We provide a new algorithm to perform the agents' computations
  on locally tree-like graphs. Our algorithm uses the dynamic
  cavity method to drastically reduce computational effort. Let
  $d$ be the maximum degree and $t$ be the iteration number. The
  computational effort needed per agent is exponential only in
  $O(td)$ (note that the number of possible information sets of a
  neighbor at time $t$ is itself exponential in $td$).

  Under appropriate assumptions on the rate of convergence, we
  deduce that each agent is only required to spend
  polylogarithmic (in $1/\eps$) computational effort to
  approximately learn the true state of the world with error
  probability $\eps$, on regular trees of degree at least
  five. We provide numerical and other evidence to justify our
  assumption on convergence rate.

  We extend our results in various directions, including loopy
  graphs. Our results indicate efficiency of iterative Bayesian
  social learning in a wide range of situations, contrary to widely held
  beliefs.
\end{abstract}

\section{Introduction}

Consider a group of Facebook users who are each faced with the dilemma
of whether to place an order for the new {\em iGadget} or the new {\em
  Gagdetoid}. Each boldly ventures to the wild and does independent
research on the subject matter, discovering the ``correct'' answer
with some probability $p>1/2$. Then, over the next few weeks, before
making the final decision, they daily share their current opinion on
the matter with their Facebook contacts by posting either {\em
  iGadget} or {\em Gadgetoid} on their status line. Every day, after
learning their friends' opinions, they update their own by performing
the Bayesian calculation that determines which of the two options is
more likely to be true, given all they know. Eventually, they make a
purchase based on this information. Such dynamics have become an
integral part of {\em electronic commerce}, and understanding them is
valuable to social media advertisers and vendors.

This model (or rather, a slightly more general version of it) was
introduced by Gale and Kariv~\cite{GaleKariv:03}. It is one in a long
succession of social learning models. Already in 1785
Condorcet~\cite{Condorcet:85} considered how a group of individuals
with weak private signals could reach a correct collective decision;
he showed that a majority vote is likely to be correct when the group
is large enough.  Models such as those of Banerjee~\cite{Banerjee:92},
Bikhchandani, Hirshleifer and Welch~\cite{BichHirshWelch:92} and Smith
and Sorensen~\cite{SmithSorensen:00} allow for each individual to make
a single decision, learning from the decisions of her
predecessors. The models of DeGroot~\cite{DeGroot:74} and Bala and
Goyal~\cite{BalaGoyal:98} consider social networks and repeated
interactions between agents.

The model of Gale and Kariv combines features from all of the above. It
describes a group of individuals (or agents), each with a private
signal that carries information on an unknown state of the world. The
individuals form a social network, so that each observes the actions
of some subset - her neighbors. The agents must choose between a set
of possible actions, the relative merit of which depends on the state
of the world. The agents iteratively learn by observing their
neighbors' actions, and picking an action that is myopically optimal,
given the information known to them.

Even in the simple case of two states of the world, two possible
private signals and two actions, the required calculations appear to
be very complicated. This has indeed been a recurring criticism of
this model (see, e.g., \cite{Jackson:08, Goyal:08}). One approach to
this difficulty is the bounded rationality approach of Bala and
Goyal~\cite{BalaGoyal:98}, where agents ignore
part
of the information available to them and perform a Bayesian
calculation on the rest.

While the bounded rationality approach has led to impressive results,
it has two disadvantages, as compared with a fully Bayesian one:
first, it is bound to involve a somewhat arbitrary decision of which
heuristics the agents use. Second, a game theoretic analysis of
strategic players is possible only if the players choose actions that are
optimal by some criterion. Hence game-theoretic analyses of learning
on networks (e.g. \cite{VieRosSol:06}) often opt for the more
difficult but fully Bayesian model.

A different approach to the difficulty of computation in the Bayesian
model is to show that the calculations are in fact not as difficult as
they appear, at least in some cases. In this paper we show that when
the graph of social ties is locally a tree, or close to one, then the
computational outlook is not a bleak as previously thought.

We first give a simple dynamic programming algorithm for the Gale and
Kariv model that is exponential in the number of individuals. Since at
iteration $t$ one may consider only agents at distance $t$, then in
graphs of maximum degree $d$ (on which we focus) the number of
individuals to consider is $O((d-1)^t)$, and the time required of each
individual to compute their action (or vote) at time $t$ is
$2^{O((d-1)^t)}$. We then develop a sophisticated dynamic program for
locally tree-like graphs that reduces the computational effort to
$2^{O(td)}$.

We conjecture, and show supporting numerical evidence, that on
infinite trees of degree at least five, the number of iterations
needed to calculate the correct answer with probability $1-\eps$ is
$O(\log\log (1/\eps))$. In fact, we rigorously establish this for the
`majority dynamics' update rule, in which agents adopt the opinion of
their neighbors in the previous round. Thus, our conjecture follows if
iterative Bayesian learning learns at least as fast as majority, as
suggested by intuition and numerical evidence, which we present.
Assuming this conjecture, the computational effort required drops from
quasi-polynomial in $1/\eps$ (using the naive dynamic program) to
polylogarithmic in $(1/\eps)$.

An additional difficulty of the Gale and Kariv model is that it
requires the individuals to exactly know the structure of the graph. A
possible solution to this is a modification that allows the agents to
know only their own neighborhoods and the distribution from which the
rest of the graph was picked. We pursue this for the natural
configuration model of random graphs (see below for full explanation)
and show that the same computational upper bounds apply here.

We also introduce two further features into the model and show how to
deal with them algorithmically. First, there may be a finite number of
`hub' nodes who are each observed by many nodes leading to several
short loops in the connectivity graph. We show that our algorithm can
be suitably modified for this case. Second, we consider that nodes may
not all be `active' in each round, and that nodes may observe only a
random subset of active neighbors.  We show that this can be handled
when `inactive' edges/nodes occur independently of each other and in
time.

The key technique used in this paper is the dynamic cavity method,
introduced by Kanoria and Montanari~\cite{KanMon:09} in their study of
``recursive majority'' updates on trees, which was also motivated by
social learning. A dynamical version of the cavity method of
Statistical Physics, this technique was used to analyze majority
dynamics on trees, and appears promising for the analysis of iterative
tree processes in general.  In this work, we use this technique for
the first time to give an algorithm for efficient computation by
nodes. This is in contrast to the case of majority updates, where the
update rule is computationally trivial.
Our algorithmic approach
leveraging the dynamic cavity method may be applicable to a range of
iterative update situations on locally treelike graphs.

\section{Model}
\begin{itemize}
\item There is a true state of the world $s \in \cS$, where $\cS$ is
  finite. The prior distribution $\P{s}$ is common knowledge.
\item Let $G=(V,E)$ be an undirected connected graph of agents and
  their social ties. Let $n\equiv |V|$.
\item Denote by $\partial i$ the neighbors of agent $i$, not including
  $i$.
\item Each agent $i$ receives a private signal $x_i \in
  \cX$, where $\cX$ is finite. Private signals are independent conditioned
  on $s$. The distribution $\CondP{x_i}{s}$ is common knowledge. We
  assume that the signal is informative, so that $\CondP{x_i}{s}$ is
  different for different values of $s$.
\item We identify the set of actions available to agents
  with the set $\cS$ of the states of the world (thus we call
  the actions `votes'). For each state of the
  world $s$, action $\sigma$ has utility one when the state of the
  world is $s=\sigma$, and zero otherwise. Thus the action that
  maximizes the expected utility corresponds to the maximum \emph{a
    posteriori} probability (MAP) estimator of the state of the world.
\item At each time period $t \in \{0, 1, 2, \ldots\}$ each agent takes an action and
  then observes the actions take by her neighbors.
\item Denote by $\mathcal{F}_i^t$ the information available to agent
  $i$ at time $t$. We do {\bf not} include in this her neighbors'
  votes at time $t$.
\item At each time period, the agents' goal is to maximize their
  expected utility. They are myopic and fully Bayesian, and so at time
  $t$ agent $i$ takes action $\argmax_{s \in
    \cS}\CondP{s}{\mathcal{F}_i^t}$, using some tie breaking rule if
  necessary. This tie breaking rule is also common knowledge.
\item $\sigma_i(t)$ denotes agent $i$'s action at time
  $t$. $\sigma_i^t=\{\sigma_i(t')|t' \leq t\}$ denotes all of
  agent $i$'s actions, up to and including time $t$. Then
  $\mathcal{F}_i^t$ includes $x_i$, $\{\sigma_j^{t-1}|j
  \in \partial i\}$ and $\sigma_i^{t-1}$ (which is actually
  a function of the first two in case of a deterministic tie breaking
  rule, see below).
\item We refer to $\sigma_i$ as $i$'s {\em trajectory}.
\item Denote $\sigma_{\partial i}^t=\{\sigma_j^t|j \in \partial
  i\}$. We assume a deterministic tie-breaking rule, so that
   $\sigma_i^t$ is a deterministic function of $x_i$ and
  $\sigma_{\partial i}^{t-1}$.
  To differentiate the random variable $\sigma_i^t$ from the function
  used to calculate it, we denote the function by $g_i^t$:
  \begin{equation*}
    \sigma_i^t=g_i^t(x_i,\sigma_{\partial i}^{t-1}).
  \end{equation*}
  For convenience, we also define the scalar functions $g_{i,t}(x_i,\sigma_{\partial i}^{t-1})$ corresponding to $\sigma_i(t)$, so that $g_i^t = (g_{i,0}, g_{i,1}, \ldots, g_{i,t})$.
\end{itemize}

\section{A Simple Algorithm}
\label{sec:simple_dp}

A sign of the complexity of this Bayesian calculation is that even the
brute-force solution for it is not trivial. We therefore describe it
here.

One way of thinking of the agents' calculation is to imagine that they
keep a long list of all the possible combinations of initial signals
of all the other agents, and at each iteration cross out entries that
are inconsistent with the signals that they've observed from their
neighbors up to that point. Then, they calculate the probabilities of
the different possible states of the world by summing over the entries
that have yet to be crossed out.

This may not be as simple as it seems. To understand which initial
configurations are ruled out by a signal coming from a neighbor, an
agent must ``simulate'' that neighbor's behavior, and so each agent
must calculate the function $g_i^t$ for every other agent $i$ and
every possible set of observations by $i$. We formalize this below.

Let $\ux \in \cX^n$ be the vector
  of private signals $(x_i)_{i \in V}$. The trajectory of $i$, $\sigma_i$, is a deterministic function of
$\ux$. Assume then that up to time $t-1$ each agent has calculated the
trajectory $\sigma_i^{t-1}(\ux)$ for all possible private signal vectors
$\ux$ and all agents $i$. This is trivial for $t-1=0$.

We say that $\uy$ is feasible for $i$ at time $t$ if $x_i=y_i$ and
$\sigma_{\partial i}^t=\sigma_{\partial i}^t(\uy)$.  We denote this
set of feasible private signal vectors $I_i^t(x_i,\sigma_{\partial
  i}^t)$.
  To calculate $\sigma_i^t(\ux)$, one need only note that
\begin{align*}
  \P{s|\cF_i^t}
  &\propto \P{s} \P{x_i, \sigma_{\partial i}^{t-1}|s}\\
  &=\P{s} \sum_{\uy \in I_i^{t-1}\left(y_i,\sigma_{\partial i}^{t-1}\right)}\CondP{\uy}{s}
  \label{eq:posterior_sow}
\end{align*}
and
\begin{align*}
  g_{i,t}(x_i, \sigma_{\partial i}^{t-1}) = \argmax_{s \in \cS} \P{s|\cF_i^t}
\end{align*}
by definition.  We use the standard abusive notation $\P{x_i}$ instead
of $\P{x_i=y_i}$, $\P{\sigma_j^t}$ instead of
$\P{\sigma_j^t=\omega_j^t}$, etc.

It is easy to verify that using this the calculation of each
$\sigma_i^t(\ux)$ takes $O(tn|\cX|^n)$. One can do better than
perform each of these separately, but in any case the result is
exponential in $n$, so we derive a rough upper bound of $2^{O(n)}$ for
this method. Since we are in particular interested in graphs of
maximum degree $d$, we note that up to time $t$ an agent need only
perform this for agents at distance at most $t$, and so this bound
becomes $2^{O((d-1)^t)}$ for large graphs, i.e., graphs for which
$n>(d-1)^t$ for relevant values of $t$.

\section{The Dynamic Cavity Algorithm on Trees}
Assume in this section that the graph $G$ is a tree with finite degree nodes.
For $j \in \partial i$ let $G_{j \to
  i}=(V_{j \to i},E_{j \to i})$ denote $j$'s connected component in
the graph $G$ with the edge $(i,j)$ removed. That is, $V_{j \to i}$ is
$j$'s subtree when $G$ is rooted in $i$.

\subsection{The Dynamic Cavity Method}
We consider a modified process where agent $i$ is replaced by a {\em
  zombie} which takes fixed actions
$\tau_i=(\tau_i(0),\tau_i(1),\ldots)$, and {\em the true state of the
  world is assume to be some fixed $s$}. Furthermore, this `fixing'
goes unnoticed by the agents (except $i$, who is a zombie anyway) who carry on their calculations,
assuming $i$ is her regular Bayesian self, and the state of the world
is drawn randomly according to $\P{s}$. We denote by
$\ZombP{A}{\tau_i,s}$ the probability of event $A$ in this modified
process. This modified process is easier to analyze, as the processes
on each of the subtrees $V_{j \to i}$ are independent. This is
formalized in the following claim, without proof:
\begin{claim}
  \begin{equation}
    \label{eq:q_independent_sw}
    \ZombP{\sigma_{\partial i}^t}{\tau_i,s}
    =\prod_{j \in \partial i}\ZombP{\sigma_j^t}{\tau_i^t,s}.
  \end{equation}
\end{claim}
(Since $\sigma_j^t$ is unaffected by $\tau_i(t')$ for all $t'>t$,
we only need to specify $\tau_i^t$, and not the entire $\tau_i$.)

Now, it might so happen that for some number of steps the zombie behaves
exactly as may be expected of a rational player. More precisely, given
$\sigma_{\partial i}^{t-1}$, it may be the case that  $\tau_i^t=g_i^t\left(x_i,\sigma_{\partial
    i}^{t-1}\right)$. This event provides the connection between the
modified process and the original process, and is the inspiration for
the following theorem.

\begin{theorem}
  \label{thm:p_eq_q}
  For all $i$, $t$ and $\tau_i$
  \begin{align}
    &\CondP{\sigma_{\partial i}^{t-1}}{s,x_i}
    \Ind{\tau_i^t=g_i^t\left(x_i,\sigma_{\partial i}^{t-1}\right)}
    = \nonumber \\
    &\phantom{\CondP{a}{b}}\ZombP{\sigma_{\partial i}^{t-1}}{\tau_i,s}
    \Ind{\tau_i^t=g_i^t\left(x_i,\sigma_{\partial i}^{t-1}\right)}\,.
    \label{eq:p_eq_q}
  \end{align}
\end{theorem}
\begin{proof}
  We couple the original process, after choosing $s$, to the modified
  processes by setting the private signals to be identical in both.

  Now, clearly if it so happens that
  $\tau_i^t=g_i^t\left(x_i,\sigma_{\partial i}^{t-1}\right)$ then the
  two processes will be identical up to time $t$. Hence the
  probabilities of events measurable up to time $t$ will be identical
  when multiplied by $\Ind{\tau_i^t=g_i^t\left(x_i,\sigma_{\partial
        i}^{t-1}\right)}$, and the theorem follows.
\end{proof}

Using Eqs.~\eqref{eq:q_independent_sw} and \eqref{eq:p_eq_q},
we can easily write the posterior on $s$ computed by node $i$ at time $t$,
in terms of  the probabilities $\ZombP{\cdot}{\cdot}$:
\begin{align}
  \P{s|\cF_i^t}
  &\propto \P{s} \P{x_i, \sigma_{\partial i}^{t-1}|s} \nonumber \\
  &= \P{s} \P{x_i|s} \P{\sigma_{\partial i}^{t-1}|s,x_i}\nonumber \\
  &= \P{s} \P{x_i|s} \prod_{j \in \partial
    i}\ZombP{\sigma_j^{t-1}}{\sigma_i^{t-1},s}
  \label{eq:posterior_sow}
\end{align}
(Note that $\sigma_i^{t-1}$ is a deterministic function of $(x_i,
\sigma_{\partial i}^{t-1})$.)

Given that, the decision function is as before
\begin{align}
  g_{i,t}(x_i, \sigma_{\partial i}^{t-1}) = \argmax_{s \in \cS} \P{s|\cF_i^t}
  \label{eq:decision_sow}
\end{align}
As mentioned before, we assume there is a deterministic tie breaking
rule that is common knowledge.

We are finally left with the task of calculating
$\ZombP{\cdot}{\cdot}$.  The following theorem is the heart of the
dynamic cavity method and allows us to perform this calculation:
\begin{theorem}
\label{thm:cavity_recursion}
   For $j \in \partial i$ and  $t \in \N$
  \begin{align}
    \label{eq:recursion_sw}
    &\ZombP{\sigma_j^t}{\tau_i,s}= \nonumber\\
    &\sum_{\sigma_{1}^{t-1} \ldots \, \sigma_{d-1}^{t-1}}
      \sum_{x_j}\, \P{x_j|s}
          \ind \left [\sigma_j^t=g_j^t\left(x_j,(\tau_i^{t-1}, \sigma_{\partial j \backslash i}^{t-1})\right) \right ]
          \, \cdot \nonumber\\
    & \phantom{xxxxxxxxxxx} \cdot \, \prod_{l=1}^{d-1}
          \ZombP{\sigma_{l}^{t-1}}{\sigma_j^{t-1},s} \, .
  \end{align}
  where the neighbors of node $j$ are $\partial j=\{i,1,2, \ldots, d-1\}$.
\end{theorem}
We mention without proof that the recursion easily generalizes to the
case of a {\em random} tie-breaking rule, provided the rule is common
knowledge; it is a matter of replacing the expression $\ind \left
  [\sigma_j^t=\cdots\right ]$ with
$\P{\sigma_j^t=\cdots}$, where this probability is over
the randomness of the rule.  Eq.~\eqref{eq:posterior_sow} continues to
be valid in this case.

The following proof is similar to the proof of Lemma 2.1
in~\cite{KanMon:09}, where the dynamic cavity method is introduced and
applied to a different process.
\begin{proof}
  In the modified process the events in the different branches that
  $i$ sees are independent. We therefore consider $V_{j \to i}$ only,
  and view it as a tree rooted at $j$. Also, for convenience we define
  $\sigma_i^t \equiv \tau_i^t$; note that the random variable
  $\sigma_i^t$ does not exist in the modified process, as $i$'s
  trajectory is fixed to $\tau_i$.

  Let $\ux$ be the vector of private signals of $j$ and all the
  vertices up to a distance $t$ from $j$ (call this set of vertices
  $V_{j \rightarrow i}^t$).  For each $l\in\{1,\dots, d-1\}$, let
  $\ux_l$ be the vector of private signals of $V_{l \to j}^{t-1}$.
  Thus, $\ux=(x_j, \ux_1, \ux_2, \ldots, \ux_{d-1})$.

  The trajectory $\sigma_j^{t}$ is a function - deterministic, by our
  assumption - of $\ux$ and $\tau_i^{t}$.  We shall denote this
  function by $F_{j\rightarrow i}$ and write $\sigma_{j}^{t}=
  F_{j\rightarrow i}^{t}(\ux,\tau_i^{t})$.  This function is uniquely
  determined by the update rules $g_l^t\left(x_l,\sigma_{\partial
      l}^{t-1}\right)$ for $l \in V_{j \rightarrow i}^t$.

  We have therefore
  \begin{eqnarray}
    \ZombP{\sigma_j^{t}=\traj^{t}}{\tau_i^{t},s} = \sum_{\ux}
    \CondP{\ux}{s} \ind \!\left(
      \traj^{t} = F_{j \rightarrow i}^{t}(\ux,\tau_i^{t})\right) .
    \label{eq:recursion_defn}
  \end{eqnarray}
  We now analyze each of the terms appearing in this sum.  Since the
  initialization is i.i.d., we have
  \begin{eqnarray}
    \CondP{\ux\,}{s}=
    \CondP{x_j}{s}\CondP{\ux_1}{s} \CondP{\ux_2}{s} \ldots
    \CondP{\ux_{d-1}}{s}\, .\label{eq:recursion_prob}
  \end{eqnarray}

  The function $F^{t}_{j \rightarrow i}(\cdots)$ can be decomposed as
  follows:
  \begin{align}
&    \ind \!\left(\traj^{t} = F^{t}_{j \rightarrow i}
      (\ux,\tau_i^{t})\right) =
    \sum_{\sigma_1^{t-1}\dots\sigma_{d-1}^{t-1}}
    \!\! \ind\!\left(\traj^t =
      g_j^t(x_{j},\sigma_{\partial j}^{t-1})\right)
 \nonumber\\
& \phantom{xxxxxxxxxxx}    \cdot \prod_{l=1}^{d-1}\ind \!\left(\sigma_l^{t-1} = F^{t-1}_{l
        \rightarrow j}(\ux_l, \traj^{t-1})\right) .
    \label{eq:recursion_ind}
  \end{align}
  Using Eqs.~(\ref{eq:recursion_prob}) and (\ref{eq:recursion_ind}) in
  Eq.~(\ref{eq:recursion_defn}) and separating terms that depend only
  on $\ux_i$, we get
  \begin{align*}
&    \ZombP{\sigma_j^{t} = \traj^{t}}{ \tau_i^{t},s} =\\[3pt]
 &
\sum_{\sigma_{1}^{t-1} \ldots \, \sigma_{d-1}^{t-1}}
    \sum_{x_j}\P{x_j|s} \ind \left(\traj^t = g_j^t(x_{j},\sigma_{\partial j}^{t-1}\right) \, \cdot\\
    & \phantom{====} \cdot \, \prod_{l=1}^{d-1} \sum_{\ux_l}
    \CondP{\ux_l}{s} \; \ind \left(\sigma_l^{t-1} = F^{t-1}_{l
        \rightarrow j}(\ux_l, \traj^{t-1})\right) \, .
  \end{align*}
  The recursion follows immediately by identifying that the product
  over $l$ in fact has argument
  $\ZombP{\sigma_{l}^{t-1}}{\sigma_j^{t-1},s}$.
\end{proof}

\subsection{The Agents' Calculations}

We now have in place all we need to perform the agent's
calculations. At time $t=0$ these calculations are trivial. Assume
then that up to time $t$ each agent has calculated the following
quantities:
\begin{enumerate}
\item $\ZombP{\sigma_j^{t-1}}{\tau_i^{t-1},s}$, for all $s$
  and for all $i,j \in V$ such that $j \in \partial i$, and for all
  $\tau_i^{t-1}$ and $\sigma_j^{t-1}$.
\item
  $g_i^t(x_i,\sigma_{\partial i}^{t-1})$ for all $i$, $x_i$ and
  $\sigma_{\partial i}^{t-1}$.
\end{enumerate}
Note that these can be calculated without making any observations -
only knowledge of the graph is needed.

At time $t+1$ each agent makes the following calculations:
\begin{enumerate}
\item $\ZombP{\sigma_j^{t}}{\tau_i^t,s}$ for all $s,i,j,\sigma_j^{t},\tau_i^t$. These can be calculated using
  Eq.~\eqref{eq:recursion_sw}, given the quantities from the previous
  iteration.
\item $g_i^{t+1}(x_i,\sigma_{\partial i}^t)$ for all $i$, $x_i$ and
  $\sigma_{\partial i}^t$. These can be calculated using
  Eqs.~\eqref{eq:posterior_sow} and \eqref{eq:decision_sow} and the
  the newly calculated $\ZombP{\sigma_j^{t}}{\tau_i^t,s}$.
\end{enumerate}

Since agent $j$ calculates $g_i^{t+1}$ for all $i$, then she in
particular calculates $g_j^{t+1}$. Therefore, she can use this to
calculate her next action, once she observes her neighbors' actions.
A simple calculation yields the following lemma.

\begin{lemma}
  In a tree graph $G$ with maximum degree $d$, the agents can
  calculate their actions up to time $t$ with computational effort
  $n2^{O(td)}$.
\end{lemma}

In fact, each agent does not need to perform calculations for the
entire graph. It suffices for node $i$ to calculate quantities up to
time $t'$ for nodes at distance $t-t'$ from node $i$ (there are at
most $(d-1)^{t-t'}$ such nodes). A short calculation yields an
improved bound on computational effort.
\begin{theorem}
\label{thm:cavity_computational_effort}
  In a tree graph $G$ with maximum degree $d$, each agent can
  calculate her action up to time $t$ with computational effort
  $2^{O(td)}$.
\end{theorem}

\subsection{Dynamic Cavity Algorithm: Extensions}

Our algorithm  admits several extensions that we explore in this section:
Section \ref{subsubsec:loops_hubs}
discusses graphs with loops and `hubs', Section \ref{subsubsec:random_graphs}
discusses random graphs, Section \ref{subsubsec:distribution_on_graph} relaxes
the assumption that the entire graph is common knowledge and Section \ref{subsubsec:random_subsets_neighbors} allows nodes/edges to be inactive
in some rounds.

First we mention some straightforward generalizations:

It is easy to see that dynamic cavity recursion (Theorem \ref{thm:cavity_recursion})
does not depend on any special properties of the Bayesian update rule.
The update rule $g_{i,t}(\cdot)$ can be arbitrary. Thus, if agent $i$ wants
to perform a Bayesian update, he can do so (exactly) using our approach even
if his neighbor, agent $j$, is using some other update rule.

\begin{remark}
The dynamic cavity recursion can be used to enable computations of agents
even if some of them are using arbitrary update rules (provided the rules are
`well specified' and common knowledge).
\end{remark}
For instance, our approach should be applicable in `partial Bayesian' settings.

Our algorithm is easily modified for the case of a general
  finite action set $\cA$ that need not be the same as $\cS$,
  associated with a payoff function $u:\cA \times \cS \rightarrow
  \R$. Moreover the action set and payoff function can each be player
  dependent ($\cA_i$, $u_i$ respectively).

We already mentioned that there is a simple generalization to
  the case of random tie breaking rules (that are common knowledge).

Instead of having only undirected edges (corresponding to
  bidirectional observations), we can allow a subset of the edges of
  the tree to be directed. In this case, the same algorithm works with
  suitably defined neighborhood sets $\partial i$.
In other words, our result holds for the class of
  directed graphs lacking cycles of length greater than two (which
  correspond to undirected edges).

\subsubsection{Loops and Hub nodes}
\label{subsubsec:loops_hubs}

 A class of graphs that are not trees, but for which this dynamic
  cavity method can be easily extended is that of trees with `hub'
  nodes in addition.

  Consider then a graph that is not a tree, but can be transformed
  into a tree by the removal of some small set of nodes $\Vl \subset
  V$. 
  Then the same calculations above can
  still be performed, with a time penalty of $|\cX|^{|\Vl|}$; the
  calculation in Eq.~\eqref{eq:recursion_sw} is simply repeated for
  each possible set of private signals of the hub nodes, and the
  probabilities in Eq.~\eqref{eq:posterior_sow} are arrived at by
  averaging the $|\cX|^{|\Vl|}$ different possible cases. In fact, one
  may not even need to average over all nodes in $\Vl$, since at iteration $t$
  only those inside $B^t_i$ (the ball of radius $t$ around $i$) effect
  the outcome of $i$'s calculations. Hence the complexity of
  calculations up to iteration $t$ is now $|\cX|^{n_t} 2^{O(td)}$,
  where $n_t = \max (|B^t_i \cap \Vl|)_{i \in V}$.

  If we also allow directed edges in this model, then we can extend it
  to include nodes of unlimited in-degree, i.e., some nodes may be observed by a unbounded
  number of others. These are agents who are
  observed by any number (perhaps an infinity) of other agents, in the
  spirit of Bala and Goyal's ``royal family'' \cite{BalaGoyal:98}. We call such nodes
 `hubs' for obvious reasons. For instance, a popular blogger or a newspaper might
 constitute such a hub. Here
  too the same computational guarantees holds.


\subsubsection{Random graphs}
\label{subsubsec:random_graphs}

Consider a random graph on $n$ nodes drawn from the configuration model with a given degree distribution\footnote{In the configuration model, one first
  assigns a degree to each node, draws the appropriate number of
  `half-edges' and then chooses a uniformly random pairing between
  them. One can further specify that a graph constructed thus is
  `rejected' if it contains double edges or self-loops; this does not
  change any of the basic properties, e.g., the local description, of
  the ensemble.} It is well known that such graphs are locally tree-like
  with high probability(see, e.g. \cite{DemboMontanari:09}). More
  formally, for any $t < \infty$ we have
  \begin{align}
  \lim_{n \rightarrow \infty} \P{B_i^t \mbox{ is a tree.}} = 1 \, .
  \end{align}
  Since node calculations up to time $t$ depend only on $B_i^t$, it
  follows that with high probability (w.h.p.), for an arbitrarily selected node, the tree
  calculations suffice for any constant number of
  iterations.\footnote{In fact, as mentioned earlier, nodes with a
    small number of loops in the vicinity can also do their
    calculations without trouble.} As we show in Section
  \ref{sec:rapid_convergence}, just $O(\log \log 1/\eps)$ iterations
  (a small number independent of $n$) are enough to learn the true
  state of the world with probability at least $1-\eps$ for any
  $\eps>0$, provided private signals are not too noisy. Thus, our
  computational approach works for random graphs w.h.p.

\subsubsection{Learning without Knowledge of the Graph}
\label{subsubsec:distribution_on_graph}
Here we consider the situation where nodes do not know the
actual graph $G$, but know some distribution over possibilities for
$G$. This is potentially a more realistic model. Moreover, the assumption
that agents are assumed to know the entire graph structure is considered
 a weakness of the model of Gale and Kariv. We address this
issue here, showing that our algorithm can be modified to allow
Bayesian estimation in this case as well.

Let $G \equiv G_n$ be a random graph of $n$ nodes constructed
according to the configuration model for a given (node perspective
degree) distribution. Denote the degree distribution by $\Dv$, so that
$\Dv(d) \equiv$ probability that a randomly selected node has degree
$d$.

Now, in this ensemble, the local neighborhood up to distance $D$ of an
arbitrary node $v$ with fixed degree $d_v$ converges in distribution
as $n \rightarrow \infty$ to the following: Each of the neighbors of
node $v$ has a degree drawn independently according to the `edge
perspective' degree distribution $\De$, defined by:
\begin{align*}
\De(d) = \frac{d\Dv(d)}{\sum_{d' \in \N} d'\Dv(d')}
\end{align*}

Further, each of the neighbors of the neighbors (except $v$ itself)
again have a degree drawn independently according to $\De(d)$, and so
on up to depth $D$. Call the resulting distribution over trees $T^D_{d_v}$.

Now suppose that agents are, in fact, connected in a graph drawn from
the ensemble $G_n$ with degree distribution $\Dv$. Suppose that each
node $u$ knows the distribution $\Dv$ and its own degree $d_u$, but
does not know anything else about $G_n$.\footnote{Other `knowledge'
  assumptions can be similarly handled, for instance where a node
  knows its own degree, the degree of its neighbors and $\Dv$.} Further,
suppose that this is common knowledge.  Now in the limit $n
\rightarrow \infty$, an exact Bayesian calculation for a node $v$ up
to time $t$ depends on $\Dv$ via $T^t_{d_v}$. Since nodes know
only their own degree, there are only $\Delta$ different `types' of nodes,
where $\Delta$ is the size of the support of $\De(d)$. There is one
type for each degree. This actually makes computations slightly simpler
than in an arbitrary known graph.


Fix state $s$. Take an arbitrary node $i$. Make it a `zombie'
following the vote trajectory $\tau_i$. Now fix some $\partial i$
(ensure $\Dv(|\partial i|)>0$). Choose arbitrary $j \in \partial
i$. Define $\ZombP{\sigma_j^t = \omega_j^t}{\tau_i^t,s}$ as the
probability of seeing trajectory $\sigma_j^t = \omega_j^t$ at node $j$
in this setting. This probability is over the graph realization
(given $\partial i$) and over the private signals.
Note here that {\em $\ZombP{\sigma_j^t = \omega_j^t}{\tau_i^t,s}$ is
  the same for any $i$, $\partial i$ and $j \in \partial i$}.

Eqs.~\eqref{eq:q_independent_sw}, \eqref{eq:posterior_sow} and
\eqref{eq:decision_sow} continue to hold w.h.p.\footnote{We need the ball of radius $t$ around $i$ to be a tree.} for the same reasons as before.

The dynamic cavity recursion, earlier given by Eq.~\eqref{eq:recursion_sw}, becomes
\begin{align}
    \label{eq:recursion_sw_unknowngraph}
    \ZombP{\sigma_j^t}{\tau_i,s}&= \sum_{d \in \N}\, \De(d)
      \sum_{\sigma_{1}^{t-1} \ldots \, \sigma_{d-1}^{t-1}}
      \sum_{x_j}\, \P{x_j|s} \, \cdot \nonumber\\
&\cdot \,          \ind \left [\sigma_j^t=g_j^t\left(x_j,(\tau_i^{t-1}, \sigma_{\partial j \backslash i}^{t-1})\right) \right ]
          \; \cdot \nonumber\\
&         \cdot \; \prod_{l=1}^{d-1}
          \ZombP{\sigma_{l}^{t-1}}{\sigma_j^{t-1},s} \, .
\end{align}
 We have written the recursion assuming the neighbors of $j$ are named according to $\partial j \backslash i= \{1, 2, \ldots, d-1\}$. Again, this holds w.h.p.\ with respect to $n$.

We comment that there is a straightforward generalization to the case
of a multi-type configuration model with a finite number of
types. Nodes may or may not be aware of the type of each of their
neighbors (both cases can be handled).  For instance, here is a simple
example with two types: There are `red' agents and `blue' agents, and
each `red' agent is connected to 3 `blue' agents, whereas each `blue'
agent is connected to either 5 or 6 `red' agents with equal
likelihood. In this case the degree distribution itself ensures that
nodes know the type of their neighbors as being the opposite of their
own type.
Multi-type configuration models are of interest since they allow for
a rich variety `social connection' patterns.

\subsubsection{Observing random subsets of neighbors}
\label{subsubsec:random_subsets_neighbors}

We may not interact with each of our friends every day.
Suppose that for each edge $e$, there is a probability $p_e$ that the
edge will be `active' in any particular iteration, independent of
everything else. Let $a_e(t) \in \{*, {\tt a}\}$, be an indicator
variable for whether edge $e$ was active at time $t$ ($\tt a$ denotes
`active').  Now, the observation by node $i$ of node $j$ belongs to an
extended set that includes an additional symbol $*$ corresponding to
the edge being inactive. Thus, there are $(|\cS|+1)^{t+1}$ possible
observed trajectories up to time $t$. Our algorithm can be easily
adapted for this case. The modified `zombie' process involves fixing
state of the world $s$, trajectory $\tau_i$ and also $(a_{ij}(t))_{j
  \in \partial i}$ for all times $t$. The form of posterior on the
state of the world, Eq.~\eqref{eq:posterior_sow}, remains
unchanged. The cavity recursion Eq.~\eqref{eq:recursion_sw} now
includes a summation over the possibilities for $(a_1^{t-1}, \ldots,
a_{d-1}^{t-1})$. The overall complexity remains $2^{O(td)}$.

The case where node $v$ becomes inactive with some probability $p_v$ in an iteration, independent
of everything else, can also be handled similarly. A suitable formulation can also be obtained
when both the above situations are combined, so that both nodes and edges may be
inactive in an iteration.

\section{Rapid learning on trees}
\label{sec:rapid_convergence}

We say that there is {\em doubly exponential convergence} to the state of the world $s$ if the error probability $\P{\sigma_i(t)\neq s}$ decays with round number $t$ as
\begin{align}
-\log(\P{\sigma_i(t)\neq s}) \in \Omega(b^t)
\end{align}
where $b>1$ is some constant.

The following is an immediate corollary of Theorem \ref{thm:cavity_computational_effort}.
\begin{corollary}
\label{coro:polylog}
Consider iterative Bayesian learning on a tree of with maximum degree $d$. If we have doubly
exponential convergence to $s$, then computational effort that is polylogarithmic in $(1/\eps)$ suffices to achieve error probability $\P{\sigma_i(t)\neq s} \leq \eps$.
\end{corollary}

 We are handicapped by the fact that very little in known rigorously about convergence of iterative Bayesian learning. Nevertheless we provide the following evidence for doubly exponential convergence on
 trees:

In Section \ref{subsec:directed_trees}, we study a simple case with
two possible states of the world and two possible private signal
values on a regular \emph{directed} tree. We show that except for the
case of very noisy signals, we have doubly exponential convergence if the degree is at least five.

Next, in Section \ref{subsec:numerical_results} we state a conjecture and show that it implies doubly exponential convergence of
iterative Bayesian learning also on undirected trees. We provide numerical evidence in support of our conjecture.

\subsection{Directed trees}
\label{subsec:directed_trees}

Consider an infinite directed $d$-ary tree. By this we mean a tree graph where
each node $i$ has one `parent' who observes $i$ and $d$ `children'
whom $i$ observes, but who do not observe $i$. Learning in such a tree
is much easier to analyze (than an undirected tree) because the trajectories of the $d$ children are uncorrelated, given $s$.

We assume a binary state of the world $s$ and independent binary
signals that are each incorrect with probability $\delta$.

\begin{lemma}
\label{lemma:directed_tree_convergence}
On an infinite directed $d$-ary tree, the error probability (at any
node) at time $t$ is bounded above by $\delta_t$, where $\delta_0
\equiv \delta$ and we have a recursive definition
\begin{align}
    \delta_t \equiv \P{{\rm Binomial}(d, \delta_{t-1}) \geq d/2} \, .
 \end{align}
\end{lemma}

\begin{proof}
  We proceed by induction on time $t$. Clearly, the error probability
  is bounded above by $\delta_0$ at $t=0$. Suppose, $\P{\sigma_i(t)
    \neq s} \leq \delta_{t}$. Consider a node $j$ making a decision at
  time $t+1$. Let the children of $j$ be $1, 2, \ldots, d$.  Define
  $\tsigma_j$, the opinion of the majority of the children, by
  \begin{align*}
  \tsigma_j(t+1)= \sgn\left(\sum_{l=1}^d \sigma_l(t)\right),
  \end{align*}
  where $\sgn(0)$
  is arbitrarily assigned the value $-1$ or $+1$. The `error-or-not' variables
  $[\sigma_l(t) \neq s]$ are iid,
  with $\P{\sigma_l(t) \neq s} \leq \delta_{t}$
  by the induction hypothesis. Hence,
  \begin{align}
    \P{\tsigma_j(t+1) \neq s} \leq \P{{\rm Binomial}(d, \delta_{t})
      \geq d/2} = \delta_{t+1} \, .
    \label{eq:errp_bound_tsigma}
  \end{align}

  Since the agent $j$ is Bayesian, she in fact uses the information
  $(x_j, \sigma_1^t, \ldots, \sigma_d^t)$ to compute a MAP estimate
  $\sigma_j(t+1)$ of the true state of the world. Clearly,
  $\P{\sigma_j(t+1)\neq s} \leq \P{\tsigma_j(t+1) \neq s}$. Using
  Eq.~\eqref{eq:errp_bound_tsigma}, it follows that
  $\P{\sigma_j(t+1)\neq s} \leq \delta_{t+1}$. Induction completes the
  proof.
\end{proof}

It follows (by an argument similar to the one used in the proof of
theorem \ref{thm:majority_doublyexpconv} below) that we have doubly
exponential convergence to the true state of the world, if the noise
level is not too high. We obtain
\begin{align}
-\log \P{\sigma_i(t)\neq s} \in \Omega\left((d/2)^t\right).
\label{eq:errp_bound_t_directed}
\end{align}
implying that $O(\log \log (1/\eps))$ rounds suffice to reduce the error
probability to below $\eps$.

\subsection{Bayesian vs. `majority' updates}
\label{subsec:numerical_results}
\begin{table}
\begin{center}
{\small
\begin{tabular}{ l | l | l }
Round  & Bayesian & Majority \\
\hline
0 & $0.15$ & $0.15$ \\
1 & $2.7 \cdot 10^{-2}$ & $2.7 \cdot 10^{-2}$\\
2 & $7.6 \cdot 10^{-4}$ & $1.7 \cdot 10^{-3}$\\
3 & $2.8 \cdot 10^{-7}$ & $8.4 \cdot 10^{-6}$\\
4 & $1.4 \cdot 10^{-12}$ & $2.5 \cdot 10^{-10}$\\
\end{tabular}
}
\end{center}
\caption{\small Error probability on regular tree with $d=5$ and
  $\P{x_i\neq s}=0.15$, for (i) Bayesian and (ii) majority
  updates. The agents break ties by picking their original private signals.}
\label{table:numerics}
\end{table}

We conjecture that iterative Bayesian learning leads to lower error
probabilities (in the weak sense) than a very simple
alternative update rule we call ``majority dynamics''. Under this rule
the agents adopt the action taken by the majority of their neighbors
in the previous iteration (this is made precise in Definition \ref{def:majority_dynamics}). Our conjecture is natural since the
iterative Bayesian update rule chooses the vote in each round that
(myopically) minimizes the error probability.

\begin{conjecture}
\label{conj:bayesian_no_worse_than_majority}
On any regular tree with independent identically distributed private signals,
the error probability under iterative Bayesian learning
is no larger than the error probability under majority dynamics (cf. Definition
\ref{def:majority_dynamics})
after the same number of iterations.
\end{conjecture}

We use $\hsigma_i(t)$ to denote votes under the majority dynamics. 

In Appendix \ref{app:majority}, we show doubly exponential convergence for majority dynamics
on regular trees:
\begin{theorem}
\label{thm:majority_doublyexpconv}
Assume binary $s$ with uniform prior. Agents' initial votes
$\hsigma_i(0)$ are correct with probability $1-\delta$, and independent conditioned on $s$.
Let $i$ be any node in an (undirected) $d$ regular tree for $d \geq 5$. Then,
under the majority dynamics,
\begin{align}
-\log \P{\hsigma_i(t)\neq s} \in \Omega\left(\left(\half(d-2)\right)^t\right).
\label{eq:errp_bound_t}
\end{align}
when $\delta < (2e(d-1)/(d-2))^{-\frac{d-2}{d-4}}$.
\end{theorem}

Thus, if Conjecture
\ref{conj:bayesian_no_worse_than_majority} holds:
\begin{itemize}
\item We also have doubly exponential convergence for
iterative Bayesian learning on regular trees with $d \geq 5$, implying that for any $\eps > 0$, an error
probability $\eps$ can be achieved in $O(\log \log (1/\eps))$ iterations
under iterative Bayesian learning.
\item Combining with Theorem \ref{thm:cavity_computational_effort} (cf. Corollary \ref{coro:polylog}), we see that
  the computational effort that is polylogarithmic in $(1/ \eps)$ suffices to achieve error
probability $1/\eps$.
\end{itemize}

This compares favorably with the $\textup{quasi-poly}(1/\eps)$ bound
on computational effort that we can derive by combining Conjecture
\ref{conj:bayesian_no_worse_than_majority} and the simple dynamic
program described in Section \ref{sec:simple_dp}.

In table~\ref{table:numerics} we provide numerical evidence on regular
undirected trees in support of our conjecture. Further numerical
results are presented in Appendix \ref{app:numerical_results}.  All
computations are exact, and were performed using the dynamic cavity
equations.  The results are all consistent with our conjecture over
different values of $d$ and $\P{x_i \neq s}$.

\begin{figure*}
\centering
\includegraphics[scale=0.75]{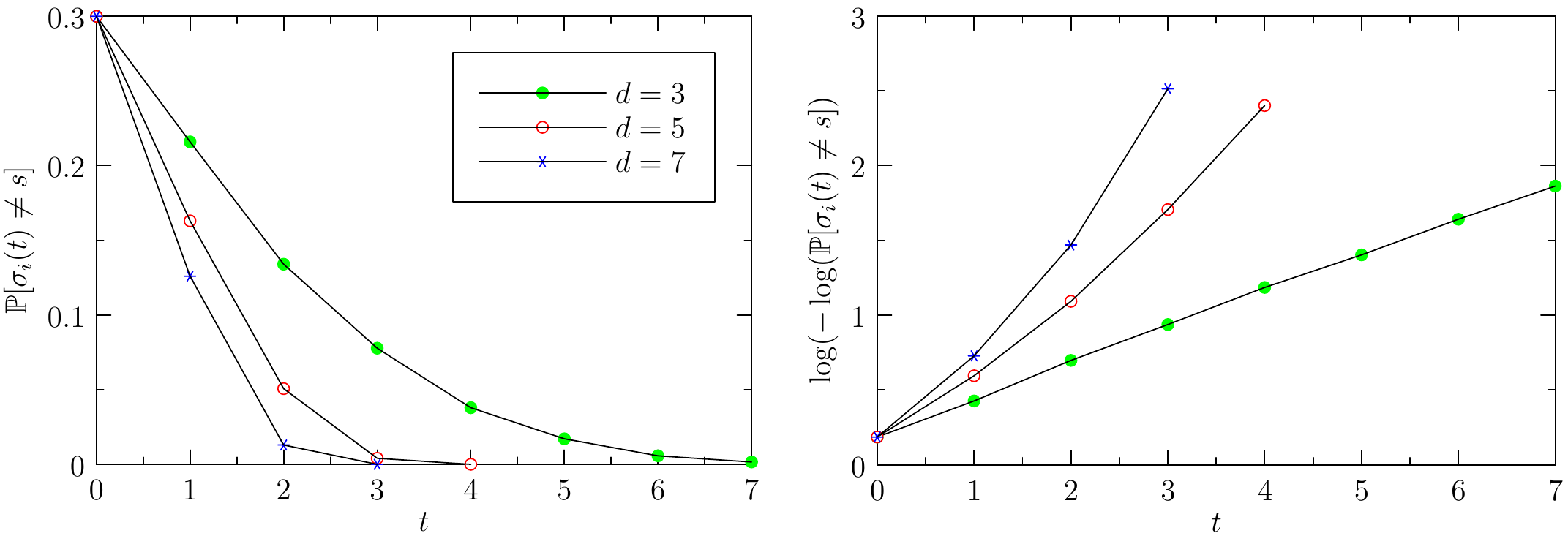}
\caption{Error probability decay on regular trees for iterative
  Bayesian learning, with $\P{x_i\neq s}=0.3$ (cf. Appendix
  \ref{app:numerical_results}). The data used to generate this figure
  is displayed in Table \ref{table:chart_data}.}
\label{fig:error_prob}
\end{figure*}

Figure \ref{fig:error_prob} plots decay of error probabilities in
regular trees for iterative Bayesian learning with $\P{x_i\neq s}=0.3$,
where the agents break ties by picking their original private
signals. Each of the curves (for different values of $d$) in the plot
of $\log (- \log \P{\sigma_i(t)\neq s})$ vs. $t$ appear to be bounded
below by straight lines with positive slope, suggesting doubly
exponential decay of error probabilities with number of iterations.

\vskip5pt
\noindent{\bf Acknowledgments.} We would like to thank Andrea Montanari,
Elchanan Mossel and Allan Sly for valuable discussions.

\bibliographystyle{abbrv} \bibliography{all}

\appendix

\section{Majority dynamics: Proof of Theorem \ref{thm:majority_doublyexpconv}}
\label{app:majority}
In this section we study a very simple update rule, `majority dynamics'.
We use $\hsigma_i(t)$ to denote votes under the majority dynamics.

\begin{definition}
\label{def:majority_dynamics}
Under the majority dynamics, each node $i \in V$ chooses his vote in round $t+1$
according to the majority of the votes of his neighbors in round $t$,
i.e.
\begin{align*}
\hsigma_i(t+1) = \textup{sign} \left ( \sum_{j \in \partial i}  \hsigma_j(t) \right )
\end{align*}
Ties are broken by flipping an unbiased coin.
\end{definition}

As before,
$s \in \{-1, +1\}$ is drawn from a 50-50 prior and nodes receive `private signals'
$\hsigma_i(0)$ that are correct with probability $1- \delta$, and independent conditioned
on $s$.

Consider an undirected $d$ regular
tree.  The analysis in this case is complicated (relative to the case of
a directed tree) by dependencies which have to be carefully handled.

\begin{lemma}
\label{lemma:undirected_tree_convergence}
  Let $i$ and $j$ be adjacent nodes in the tree. Then for all
  $(\hsigma_i^{t-1}, \hsigma_j^{t-1}) \in \{-1,+1\}^{2t}$
  \begin{align}
    \P{\hsigma_i(t)= -1|\hsigma_i^{t-1}, \hsigma_j^{t-1}, s=+1} \leq \delta_t
    \label{eq:errp_bound_t_givenij}
  \end{align}
  where $\delta_t$ is defined recursively by $\delta_0\equiv \delta$,
  and
  \begin{align}
    \label{eq:delta_t}
    \delta_t \equiv \P{{\rm Binomial}(d-1, \delta_{t-1}) \geq d/2-1}
  \end{align}
\end{lemma}

\begin{proof}
  We proceed by induction. Clearly Eq.~\eqref{eq:errp_bound_t_givenij}
  holds for $t=0$.  Suppose Eq.~\eqref{eq:errp_bound_t_givenij} holds
  for some $t$. We want to show
  \begin{align}
    \P{\hsigma_i(t+1)= -1\,|\,\hsigma_i^{t}, \hsigma_j^{t}, s=+1} \leq
    \delta_{t+1} \, ,
    \label{eq:errp_bound_tplus1}
  \end{align}
  for all $(\hsigma_i^{t}, \hsigma_j^{t}) \in \{-1,+1\}^{2(t+1)}$.

  Let $l_1, l_2, \ldots, l_{d-1}$ be the other neighbors of node $i$
  (besides $j$). We will show that, in fact,
  \begin{align}
    &\P{\hsigma_i(t+1)= -1\,|\,\hsigma_i^{t}, \hsigma_j^{t},
      \hsigma_{l_1}^{t-1}, \ldots,
      \hsigma_{l_{d-1}}^{t-1}, s=+1} \ \ \ \ \nonumber \\[3pt]
       &\phantom{\P{\hsigma_i(t+1)= -1\,|\,\hsigma_i^{t}, \hsigma_j^{t}, s=+1\ \ \ \,}}\leq \delta_{t+1} \, ,
    \label{eq:errp_givenxi}
  \end{align}
  for all possible $\xi\equiv (\hsigma_i^{t}, \hsigma_j^{t},
  \hsigma_{l_1}^{t-1}, \hsigma_{l_2}^{t-1}, \ldots,
  \hsigma_{l_{d-1}}^{t-1})$.

  We reason as follows.  Fix the state of the world $s$ and the
  trajectories $\hsigma_i^t$ and $\hsigma_j^t$. Now this induces
  correlations between the trajectories of the neighbors $l_1, \ldots,
  l_{d-1}$, caused by the requirement of consistency, but {\em only up
    to time $t-1$}. If we further fix $\hsigma_{l_m}^{t-1}$, then
  $\hsigma_{l_m}(t)$ (and $\hsigma_{l_m}$ at all future times) is
  conditionally independent of $\big(\hsigma_{l_{m'}}^t \big)_{m' \neq
    m}$. \footnote{A more formal proof can be constructed by mirroring
    the reasoning used in the proof of Theorem
    \ref{thm:cavity_recursion}.} Thus, we have
  \begin{align*}
    &\P{\hsigma_{l_m}(t)= -1|\,\xi, s=+1} = \\
    & \phantom{xxxxxxxxx} \P{\hsigma_{l_m}(t)= -1|\, \hsigma_{l_m}^{t-1}, \hsigma_i^{t-1}, s=+1} \, ,
  \end{align*}
  and therefore, using the induction hypothesis
  \begin{align}
    \P{\hsigma_{l_m}(t)= -1|\,\xi, s=+1} \leq \delta_t
    \label{eq:errp_lm_givenxi}
  \end{align}
  for all $m \in \{1, 2, \ldots,
  d-1\}$.
  Also, the actions $\hsigma_{l_1}(t), \ldots,
  \hsigma_{l_{d-1}}(t)$ are conditionally independent of each other given $\xi,
  s=+1$. We have
  \begin{align*}
    {\hsigma}_i(t+1) = \sgn(\hsigma_j(t) + \hsigma_{l_1}(t) + \ldots + \hsigma_{l_{d-1}}(t))\, ,
  \end{align*}
  with $\sgn(0)$ being assigned value $-1$ or $+1$ with equal probability. We have
  \begin{align*}
    &\P{{\hsigma}_i(t+1) = -1\,|\, \xi, s=+1} \,\leq \\
    &\phantom{xxxxxxxxxx}\P{{\rm Binomial}(d-1, \delta_{t}) \geq d/2-1}\,
  \end{align*}
  from Eq.~\eqref{eq:errp_lm_givenxi} and conditional independence of
  $\hsigma_{l_1}(t), \ldots, \hsigma_{l_{d-1}}(t)$.
    This yields
  Eq.~\eqref{eq:errp_givenxi}. Eq.~\eqref{eq:errp_bound_tplus1}
  follows by summing over $\hsigma_{l_1}^{t-1}, \hsigma_{l_2}^{t-1},
  \ldots, \hsigma_{l_{d-1}}^{t-1}$.
\end{proof}

\begin{proof}[Proof of Theorem \ref{thm:majority_doublyexpconv}]
  By applying the multiplicative version of the Chernoff
  bound\footnote{$\P{X \geq (1+\eta)\E{X}} \leq \left(\frac{\exp
        \eta}{(1+\eta)^{1+\eta}}\right)^{\E{X}}$. We substitute $\E{X}=\delta_t(d-1)$ and $1+\eta=(d/2-1)/[\delta_t(d-1)]$.} to
  Eq.~\eqref{eq:delta_t} we have that
  \begin{align*}
    \delta_{t+1} &\leq e^{(d-2)/2 -(d-1)\delta_t}(2\delta_t (d-1)/(d-2))^{(d-2)/2}
  \end{align*}
  Dropping the term
  $e^{ -(d-1)\delta_t}$, we obtain
  \begin{align}
  \delta_{t+1} \leq (2e\delta_t (d-1)/(d-2))^{\half(d-2)}.
  \end{align}

  This is a first order non-homogeneous linear recursion in $\log
  \delta_t$. If it were an equality it would yield
  \begin{align*}
  &\log \delta_t = \\
  &\ \left(\log \delta + \frac{d-2}{d-4}\log [ 2e (d-1)/(d-2)]\right)
  \left[\half(d-2)\right]^t \\
  &\ \ \ - \, \frac{d-2}{d-4}\log  [ 2e (d-1)/(d-2)] \, ,
  \end{align*}
  and so
  \begin{align}
    \label{eq:delta_asympt}
    - \log \delta_t \in \Omega \left(\left(\half(d-2)\right)^t\right) ,
  \end{align}
  as long as
  $$\log \delta < \frac{d-2}{d-4}\log [ 2e (d-1)/(d-2)]\, .$$
\end{proof}

Theorem \ref{thm:majority_doublyexpconv} is non-trivial for $d \geq 5$.  The upper limit of the noise for which it establishes rapid
convergence approaches $(2e)^{-1}$ as $d$ grows large (see also the discussion below
for large $d$).

\subsection{Convergence for large $d$}
We present now a short informal discussion on the limit $d \rightarrow
\infty$.  We can, in fact,
use Lemma \ref{lemma:directed_tree_convergence} to show
convergence is doubly exponential for $\delta < 1/2 - c/d$ for
some $c< \infty$ that does not depend on $d$.

Here is a sketch of the argument. Suppose $\delta = 1/2 -
c_1/d$. Then, for all $d>d_1$ where $d_1 < \infty$, there exists $c_2
< \infty$ such that $\P{{\rm Binomial}(d-1, \delta) \geq d/2-1} < 1/2
- c_2/\sqrt{d}$. This can be seen, for instance, by coupling with the
${\rm Binomial}(d-1, 1/2)$ process and using an appropriate local
central limit theorem (e.g., see \cite[Theorem 4.4]{KanMon:09}). Thus,
$\delta_1 < 1/2 - c_2/\sqrt{d}$.  Further, $c_2$ can be made
arbitrarily large by choosing large enough $c_1$. Next, with a simple
application of the Azuma's inequality, we arrive at $\delta_2 < c_3$
(where $c_3 \rightarrow 0 $ as $c_2 \rightarrow \infty$). Now, for
small enough $c_3$, we use the Chernoff bound analysis in the proof of
Theorem \ref{thm:majority_doublyexpconv} and obtain doubly exponential
convergence.

\section{Further numerical results}
\label{app:numerical_results}
Table \ref{table:3_pt3},
together with table \ref{table:numerics} above, contrast the error
probabilities of Bayesian updates with those of majority updates.
All cases exhibit lower error probabilities (in the weak sense) for the Bayesian
update, consistent with Conjecture
\ref{conj:bayesian_no_worse_than_majority}. Table \ref{table:chart_data} contains the data plotted in
Figure \ref{fig:error_prob}. Also for these parameters, we found that the Bayesian
updates showed lower error probabilities than the majority updates (though
we omit to present the majority results here).

The running time to generate these tables, on a standard desktop machine
was less than a minute. We did not proceed with more rounds because of
numerical
instability issues which begin to appear as error probabilities decrease.
\newpage
\phantom{xxxxxxxxx}

\begin{table}[t]
\begin{center}
{\small
\begin{tabular}{ l | l | l }
Round  & Bayesian & Majority \\
\hline
0 & $0.15$ & $0.15$ \\
1 & $ 6.1 \cdot 10^{-2}$ & $6.1 \cdot 10^{-2}$\\
2 & $ 1.5 \cdot 10^{-2}$ & $3.0 \cdot 10^{-2}$\\
3 & $ 3.0 \cdot 10^{-3}$ & $1.6 \cdot 10^{-2}$\\
4 & $ 3.4 \cdot 10^{-4}$ & $9.2 \cdot 10^{-3}$\\
5 & $ 2.7 \cdot 10^{-5}$ & $5.5 \cdot 10^{-3}$\\
6 & $ 2.2 \cdot 10^{-6}$ & $3.4 \cdot 10^{-3}$\\
7 & $ 1.4 \cdot 10^{-7}$ & $3.4 \cdot 10^{-3}$
\end{tabular}
}
\end{center}
\vskip-5pt
\caption{\small $d=3$, $\P{x_i \neq s}=0.15$}
\label{table:3_pt3}
\end{table}

\begin{table}[t]
\begin{center}
{\small
\vskip10pt
\begin{tabular}{ l | l | l | l}
Round  & $d=3$ & $d=5$ & $d=7$ \\
\hline
0 & $ 0.30 $ & $ 0.30$ & $0.30$ \\
1 & $ 0.22 $ & $ 0.16$ & $0.13$ \\
2 & $ 0.13 $ & $ 5.1 \cdot 10^{-2}$ & $ 1.3 \cdot 10^{-2}$\\
3 & $ 7.8 \cdot 10^{-2}$ & $ 4.1 \cdot 10^{-3}$ & $ 4.4 \cdot 10^{-6}$\\
4 & $ 3.8 \cdot 10^{-2}$ & $ 1.6 \cdot 10^{-5}$ &\\
5 & $ 1.7 \cdot 10^{-2}$ &\\
6 & $ 5.7 \cdot 10^{-3}$ &\\
7 & $ 1.5 \cdot 10^{-3}$ &
\end{tabular}
}
\end{center}
\vskip-5pt
\caption{\small Error probabilities with $\P{x_i \neq s}=0.3$,
for regular tree of different degrees $d$. This data is
  displayed in Figure \ref{fig:error_prob}.}
\label{table:chart_data}
\end{table}

\end{document}